\documentclass[11pt,a4paper]{amsart}
\usepackage[T1]{fontenc}
\usepackage{mathtools,bbm,amsthm,amsmath,amsfonts,amssymb,color,epic}
\usepackage{multirow,comment}
\usepackage[export]{adjustbox}
\usepackage[ruled]{algorithm2e}
\usepackage{enumitem}
\usepackage{subfig}
\usepackage{float}
\usepackage{xcolor}
\usepackage{graphicx}
\usepackage{booktabs}
\definecolor{SUPSI-Blu}{HTML}{002AA7}
\definecolor{SUPSI-Ciano}{HTML}{0090C8}
\definecolor{SUPSI-Verde}{HTML}{74C244}
\definecolor{SUPSI-Antracite}{HTML}{3B3C48}
\usepackage[
    colorlinks=true,
    linkcolor=SUPSI-Ciano,
    citecolor=SUPSI-Blu,
    urlcolor=SUPSI-Verde
]{hyperref}
\theoremstyle{plain}
\usepackage{tikz}
\newtheorem{theorem}{Theorem}[section]
\newtheorem{lemma}[theorem]{Lemma}

\usetikzlibrary{calc,arrows,intersections,positioning,arrows.meta}
\tikzset{
control/.style = {rectangle, rounded corners, dotted, draw=black, minimum height=2em, text centered},
fast/.style = {rectangle, rounded corners, dashed, draw=black, minimum height=2em, text centered},
slow/.style = {rectangle, rounded corners, draw=black, minimum height=2em, text centered},
}
\newcommand{\dt}{\frac{\mathrm{d}}{\mathrm{d}t}}
\newcommand{\dtbar}{\frac{\mathrm{d}}{\mathrm{d}\overline{t}}}
\newcommand{\VO}{V\!O_2}
\newcommand{\Ie}{I_\mathrm{e}}
\newcommand{\Gup}{G_\mathrm{up}}
\newcommand{\Gpr}{G_\mathrm{pr}}
\newcommand{\IL}{I\!L6}
\newcommand{\VL}{V\!L}
\newcommand{\SI}{S_\mathrm{I}}
\newcommand{\isdef}{\mathrel{\mathrel{\mathop:}=}}
\newcommand{\defis}{\mathrel{=\mathrel{\mathop:}}}

\allowdisplaybreaks

\usepackage[backend=biber,style=numeric-comp, sorting = none, maxbibnames=20]{biblatex}

\AtEveryBibitem{%
  \ifboolexpr{
    test {\ifentrytype{book}} 
    or
    test {\iffieldundef{doi}} 
  }
  {}
  {%
    \clearfield{url}%
    \clearfield{isbn}%
    \clearfield{issn}%
  }%
  \ifboolexpr{
    not test {\iffieldundef{doi}}
    and
    not test {\iffieldundef{url}}
  }
  {\clearfield{url}}{}
}

\title[Reduced model of physical activity effects on diabetes]
{A reduced model for the long-term effects
of physical activity on type 2 diabetes}

\author[Multerer]{Lea Multerer}
\address{Lea Multerer: SUPSI, Dalle Molle Institute for Artificial Intelligence (IDSIA), Lugano, Switzerland}
\email{lea.multerer@supsi.ch}

\author[De Paola]{Pierluigi Francesco De Paola}
\address{Pierluigi Francesco De Paola: CNR, CNR-IEIIT, Milan, Italy; CNR, CNR-IASI, Rome, Italy; Politecnico of Bari, Bari, Italy}
\email{p.depaola@phd.poliba.it}

\author[Lenatti]{Marta Lenatti}
\address{Marta Lenatti: CNR, CNR-IEIIT, Milan, Italy}
\email{martalenatti@cnr.it}

\author[Paglialonga]{Alessia Paglialonga}
\address{Alessia Paglialonga: CNR, CNR-IEIIT, Milan, Italy}
\email{alessia.paglialonga@cnr.it}

\author[Azzimonti]{Laura Azzimonti}
\address{Laura Azzimonti: SUPSI, Dalle Molle Institute for Artificial Intelligence (IDSIA), Lugano, Switzerland}
\email{laura.azzimonti@supsi.ch}
\begin{document}
\begin{abstract}
Type 2 diabetes progresses slowly and may be
reversed through lifestyle changes, but
quantifying the long-term impact of regular physical activity
remains challenging due to sparse longitudinal data.
Mechanistic models offer a powerful tool
by simulating metabolic processes over extended
timescales. However, multi-scale formulations that capture
both the short-term effects of exercise sessions and the
slow evolution of disease tend to be computationally demanding,
limiting their practical use in personalized decision support.

To address this limitation, we derived
a reduced version of a two-scale model that captures the
short- and long-term effects of physical activity
on blood glucose regulation.
By analytically averaging the short-term effects
induced by exercise, we developed a homogenized formulation
that transmits the average contribution
of physical activity to the slower glucose-insulin dynamics.
This reduction preserves the key model dynamics while
decreasing computational complexity by almost a factor 2000.
We prove that the approximation error remains bounded
and confirm the model's accuracy through a
parameter-based simulation study.

The resulting model provides a mathematically grounded
reduction that retains key physiological mechanisms
while enabling fast long-term simulations.
This substantial computational gain makes it suitable for
integration into medical decision support systems,
where it can be used to design and evaluate personalized
physical activity plans aimed at reducing the risk of type 2 diabetes.
\end{abstract}

\maketitle

\section*{Keywords}
Reduced model,
Ordinary differential equations,
Homogenization,
Physical activity,
Type 2 diabetes,
Long-term progression,
Computational efficiency.

\section{Introduction}

The progression to type 2 diabetes is asymptomatic, develops over years, and may often be reversed
with adequate lifestyle changes~\cite{diabetes_prevention_program_research_group_reduction_2002,bull_world_2020}.
Personalized recommendations on modifiable risk factors,
such as physical activity, are crucial for preventing or delaying the disease progression~\cite{rooney_global_2023}.
However, longitudinal data capturing the regular impact of physical activity over many years
remain scarce, making data-driven assessment of long-term effects infeasible in practice.
Mechanistic models offer a way to bridge
this gap by simulating the complex physiological mechanisms over long timescales.

Ordinary differential equation models (ODEs)
have been widely used to model glucose-insulin interactions, capturing effects ranging from a minute-scale~\cite{bergman_physiologic_1981,roy_dynamic_2007} to long-term dynamics spanning years or even decades~\cite{topp_model_2000,ha_mathematical_2016,de_paola_long-term_2023,de_paola_modeling_2025}.
Early minimal models focused on short-term glucose and insulin interactions, using a few
state variables~\cite{bergman_physiologic_1981, bergman_minimal_2005}.
Extensions expanded these approaches to include the short-term effects of
physical activity~\cite{roy_dynamic_2007}.
Building on these foundations, a growing range of models captures complex
metabolic and inflammatory processes across multiple time-scales~\cite{palumbo_personalizing_2018,palumbo_computational_2023,multerer_computationally_2024}.
Models for long-term disease progression integrate beta-cell dynamics that account for the long-term
interactions between glucose and insulin~\cite{topp_model_2000}. Subsequent extensions
introduce additional state variables to better capture the progression over several years or even decades~\cite{de_gaetano_mathematical_2008, ha_mathematical_2016, de_gaetano_novel_2019}. However, these long-term progression models do not incorporate the effects of physical activity on the glucose-insulin dynamics.

Recently, this gap was filled by a model capturing the short- and long-term effects of regular physical activity on
blood glucose regulation in terms of twelve ODEs on two timescales~\cite{de_paola_long-term_2023, de_paola_modeling_2025}.
This model allows for a mechanistic description of
how long-term physiological processes are influenced by the integral effects of physical activity. The model is able to estimate the long-term benefits of physical activity in a variety of conditions and includes inter-individual variability, for example in terms of response to exercise, exercise program structure and adherence rates, and individual level of risk \cite{de_paola_modeling_2025}. As such, the model may be the basis for deriving personalized
physical activity recommendations to prevent the progression to type 2 diabetes.
However, due to its multi-scale nature, repeated simulations remain computationally intensive,
limiting its practical use in decision-support.
This calls for a model reduction that preserves the mechanistic
grounding while enabling fast, scalable simulation to support clinical decision-making.

A powerful approach to reducing the complexity of multi-scale differential equation models is homogenization~\cite{bakhvalov_homogenisation_1989,cioranescu_introduction_1999,sanders_averaging_2007,pavliotis_multiscale_2008,allaire_brief_2012}.
It applies to systems with dynamics across multiple scales, such as high-frequency oscillations over time or distinct spatial patterns.
The goal is to replace the short-scale oscillations with a smoothed or averaged solution to yield an effective, simplified model that captures the macroscopic dynamics.
Common applications include material science~\cite{torquato_random_2002} and fluid dynamics~\cite{hornung_homogenization_1997}.
Although traditionally used for partial differential equations,
homogenization techniques can also be effectively applied to ODEs~\cite{verhulst_methods_2005}.

In this work, we derive a reduced version
of a model that examines the long-term effects of physical activity on blood glucose regulation~\cite{de_paola_long-term_2023,de_paola_modeling_2025},
applying principles from homogenization theory.
This reduced model preserves the essential long-term influence of
physical activity on type 2 diabetes progression while being computationally efficient. We prove the boundedness of the approximation error
and inspect it in a numerical simulation study.
Our approach leads to a fast approximation of the full model,
enabling the efficient prediction of type 2 diabetes progression as a function of physical activity
and facilitating personalized physical activity plans for risk reduction.
To the best of our knowledge, this is the first application
of homogenization theory to mechanistic models of type 2 diabetes
progression,
linking model reduction with clinically relevant long-term simulations.

The remainder of this paper is structured as follows:
In Section~\ref{sec:full}, we present a scaled version of the full model to improve numerical stability
and establish the existence and uniqueness of solutions.
Section~\ref{sec:homogenization} details the reduction of the short-term state variables and provides a proof of error bounds.
In Section~\ref{sec:numerics}, we present a numerical comparison of the original and reduced
models across various parameters and initial conditions.
Finally, Section~\ref{sec:discussion} discusses the implications and contributions of this work.
\section{A scaled model for the effect of regular physical activity on type 2 diabetes progression}\label{sec:full}

The system of ODEs under investigation operates on two timescales.
The short-term equations describe how the glucose-insulin
regulation mechanism is influenced during a physical activity session on a minute-scale.
These dynamics exhibit a periodic pattern, repeating with
every physical activity session.
The long-term equations describe the evolution of glucose regulation over years,
parametrized at a day-scale. These dynamics are modulated by physical activity through a coupling with the short-term system.

\subsection{Model formulation}

We present a scaled, coupled version of the original model formulation for
improved numerical stability~\cite{langtangen_scaling_2016}.
The detailed technical steps to transform the original system into the scaled version managing the multiple scales
are provided in Appendix~\ref{appendix:scaling}.

We consider the system of ODEs
\begin{equation}\label{eq:y}
\dt \mathbf{y}(t,u)
\isdef \dt
\begin{bmatrix}
\mathbf{y_1}(t,u) \\
\mathbf{y_2}(t,\mathbf{y_1})
\end{bmatrix}
=
\begin{bmatrix}
\mathbf{f_1}(t, u, \mathbf{y_1}) \\
\mathbf{f_2}(t, \mathbf{y_1}, \mathbf{y_2})
\end{bmatrix}
\defis \mathbf{f}(t, u, \mathbf{y})
\end{equation}
for $0< t \leq t_\mathrm{end}$, together with the initial condition
$\mathbf{y}(0) = \mathbf{y_0}$.
The control $u(t)$ is a periodic continuation of a Heaviside function, defined for $n\in\mathbb{N}$ periods as
\begin{equation}\label{eq:control}
u(t) \isdef 
\begin{cases}
1 & \text{for } k\nu \leq t \leq k\nu + \delta,\\
0 & \text{for } k\nu + \delta < t < (k+1)\nu,
\end{cases}
\end{equation}
where 
$k = 0,\dots, n-1$ and $t_\mathrm{end} \isdef n\nu$.
The parameter $\nu$ denotes the period length of physical activity and
refers to the time between the start of two consecutive sessions,
and $\delta$ is the duration of a physical activity session.
The vector $\mathbf{y_1}$ contains the short-term equations and
comprises $5$ state variables
\[
\mathbf{y_1} \isdef [\VO, \Gpr, \Gup, \Ie, \IL]^T,
\]
that satisfy the following ODEs
\begin{equation}\label{eq:y1}
\dt \mathbf{y_1}
\isdef \dt
\begin{bmatrix}
\VO \\
\Gpr \\
\Gup \\
\Ie \\
\IL
\end{bmatrix}
=
\begin{bmatrix}
\lambda_t \theta(u(t) - \VO) \\
\lambda_t \alpha_2(\VO - \Gpr) \\
\lambda_t \alpha_4(\VO - \Gup) \\
\lambda_t \alpha_6(\VO - \Ie) \\
\lambda_t \kappa_{\mathrm{IL6}}(\VO - \IL)
\end{bmatrix}
\defis \mathbf{f_1}(t,u,\mathbf{y_1}),
\end{equation}
together with the initial conditions
$\mathbf{y_1}(0) = [0, 0, 0, 0, 0]^T$.
Details on the constant scaling parameter $\lambda_t$ and the other parameters $\theta$, $\alpha_2$, $\alpha_4$, $\alpha_6$ and $\kappa_{\mathrm{IL6}}$ are given in Appendix~\ref{appendix:parameters}.
The short-term variables $[\VO, \Gpr, \Gup, \Ie, \IL]^T$
capture the immediate physiological responses to physical activity.
Specifically, oxygen consumption ($\VO$) drives the
other variables during exercise.
The rate of increase in
hepatic glucose production ($\Gpr$) and glucose uptake by working tissues ($\Gup$)
reflect the acute demand for glucose during activity.
The rate of increase in insulin removal from circulation ($\Ie$) is enhanced by exercise, while the rate of release of Interleukin-6 ($\IL$), an anti-inflammatory protein~\cite{morettini_system_2017}, accounts for its anti-inflammatory effects.
These short-term equations are parametrized such that the numerical
solution of $\mathbf{y_1}$ increases as soon as physical activity starts and decays fast once it stops again.
This yields a periodic pattern of the solution $\mathbf{y_1}$ with period $\nu$, with all components remaining bounded within a fixed range for all $t \in [0, t_\mathrm{end}]$.

The vector $\mathbf{y_2}$ summarizes the long-term equations and comprises $7$ state variables,
\[
\mathbf{y_2} \isdef [\VL, \SI, \Sigma, \Gamma, B, I, G]^T,
\]
that satisfy the following ODEs
\begin{equation}\label{eq:y2}
\dt \mathbf{y_2}
\isdef \dt
\begin{bmatrix}
\VL \\
\SI \\
\Gamma \\
\Sigma \\
B \\
I \\
G
\end{bmatrix}
=
\begin{bmatrix}
h_{\VL}(\VL,\IL)\\
h_{\SI}(\SI,\VL)\\
h_{\Gamma}(\Gamma,G)\\
h_{\Sigma}(\Sigma,\Gamma,G)\\
h_{B}(B,\VL,\Gamma,\Sigma,G)\\
h_{I}(I,\Ie,\Gamma,\Sigma,B,G)\\
h_{G}(G,\Gup,\Gpr,\SI,I)
\end{bmatrix}
\defis \mathbf{f_2}(t,\mathbf{y_1},\mathbf{y_2}),
\end{equation}
together with the scaled initial conditions
\[
\mathbf{y_2}(0) = [0, 1, \Gamma_{0\lambda}, \Sigma_{0\lambda}, 1, 1, 1]^T.
\]
The right hand side functions are defined as follows:
\begin{align*}
h_{\VL}(\VL,\IL) &\isdef \lambda_t\kappa_{\mathrm{s}}(\IL - \VL),\\
h_{\SI}(\SI,\VL) &\isdef d_{\lambda}(\VL)\frac{\theta_{\SI} - \lambda_{\SI}\SI}{\tau_{\SI}},\\
h_{\Gamma}(\Gamma,G) &\isdef \frac{g_{\lambda}(G) - \Gamma}{\tau_{\Gamma}},\\
h_{\Sigma}(\Sigma,\Gamma,G) &\isdef \frac{s_{\lambda}(\Gamma, \Sigma, G) - \Sigma}{\tau_{\Sigma}},\\
h_{B}(B,\VL,\Gamma,\Sigma,G) &\isdef 
\frac{p_\lambda(\VL, \Gamma, \Sigma, G) - a_\lambda(\VL, G)}{\tau_{B}}
B,\\
h_{I}(I,\Ie,\Gamma,\Sigma,B,G) &\isdef r_{\lambda}(\Gamma, \Sigma, G)B  - \kappa I - \lambda_{\Ie I}\Ie,\\
h_{G}(G,\Gup,\Gpr,\SI,I) &\isdef \rho_{\lambda} + \lambda_{tG}\omega(\lambda_{\Gpr}\Gpr - \lambda_{\Gup}\Gup) - (\eta_0 + \lambda_{\SI I} \SI I)G.
\end{align*}
These functions include six auxiliary functions $d_{\lambda}$,
$g_{\lambda}$, $s_{\lambda}$, $p_\lambda$, $a_\lambda$ and $r_{\lambda}$ that are written out in detail in Equation~\ref{eq:final_aux_fun} in Appendix~\ref{appendix:scaling}.
These auxiliary functions are compositions of Hill-type functions and fractions incorporating exponential functions.
Again, a list of all the parameters introduced are listed in Appendix~\ref{appendix:parameters}.

The long-term equations capture the behavior of glucose regulation over a time span of years.
Due to the feedback structure of the system and the boundedness of the short-term variables $\mathbf{y_1}(t)$, these state variables remain bounded over time for the selected set of parameters.
The core of this subset of equations lies in the variables $[\VL, B, I, G]^T$, modeling the glucose-insulin ($G$ and $I$) negative feedback loop.
This feedback mechanism involves the action of beta cells ($B$), responsible for insulin release.
The variable $\VL$ accounts for the long-term effects of physical activity on beta cells and insulin sensitivity ($\SI$) promoted by Interleukin-6 ($\IL$). This variable bridges
the two timescales in the model.
Furthermore, the state variables $\Gamma$ and
$\Sigma$ model mechanisms that link the
effects of $B$ on the negative feedback loop between $G$ and $I$~\cite{ha_mathematical_2016}.

A visualization of all the state variables, the control and their interplay is given in Figure~\ref{fig:DePaola_Structure}.
\begin{figure}
\centering
\begin{tikzpicture}
\node[slow, text width=2em] (sigma) at (270:2) {$\Sigma$};
\node[slow, text width=2em] (gamma) at (198:2) {$\Gamma$};
\node[slow, text width=2em] (I) at (126:2) {$I$};
\node[slow, text width=2em] (G) at (54:2) {$G$};
\node[slow, text width=2em] (B) at (342:2) {$B$};
\node[slow, text width=2em, right of = G, xshift = 3em] (SI) {$\SI$};
\node[slow, text width=2em, above of = SI, yshift = 0.5em] (VL) {$\VL$};
\node[fast, text width=2em, above of = VL, yshift = 0.5em] (IL6) {$\IL$};
\node[fast, text width=2em, left of = IL6, xshift = -2em] (Gpr) {$\Gpr$};
\node[fast, text width=2em, left of = Gpr, xshift = -2em] (Gup) {$\Gup$};
\node[fast, text width=2em, left of = Gup, xshift = -2em] (Ie) {$\Ie$};
\node[fast, text width=2em, above=2.75em of $(Gup.west)!0.5!(Gpr.east)$] (VO) {$\VO$};
\node[control, text width=2em, above of = VO, yshift = 0.5em] (u) {$u$};
\draw[->,>=Stealth] (u.south) -- (VO.north);
\draw[->,>=Stealth] (VO.south) -- (Ie.north);
\draw[->,>=Stealth] (VO.south) -- (Gup.north);
\draw[->,>=Stealth] (VO.south) -- (Gpr.north);
\draw[->,>=Stealth] (VO.south) -- (IL6.north);
\draw[->,>=Stealth] (IL6.south) -- (VL.north);
\draw[->,>=Stealth] (VL.south) -- (SI.north);
\draw[->,>=Stealth] (SI.west) -- (G.east);
\draw[->,>=Stealth] (Ie.south) -- (I.north);
\draw[->,>=Stealth] (Gup.south) -- (G.95);
\draw[->,>=Stealth] (Gpr.south) -- (G.85);
\draw[->,>=Stealth] (VL.west) -- (B.85);
\draw[<->,>=Stealth] (I.east) -- (G.west);
\draw[->,>=Stealth] (gamma.north) -- (I.250);
\draw[->,>=Stealth] (gamma.east) -- (B.west);
\draw[->,>=Stealth] (gamma.south) -- (sigma.west);
\draw[->,>=Stealth] (sigma.95) -- (I.270);
\draw[->,>=Stealth] (sigma.east) -- (B.south);
\draw[->,>=Stealth] (B.175) -- (I.290);
\draw[->,>=Stealth] (G.south) -- (gamma.5);
\draw[->,>=Stealth] (G.south) -- (sigma.85);
\draw[->,>=Stealth] (G.south) -- (B.95);
\end{tikzpicture}
\caption{Structure of the full model $\mathbf{y}(t,u)$.
The control $u$ (dotted) influences the short-term state variables grouped together in $\mathbf{y_1}(t,u)$ (dashed) which, in turn, influence
the long-term state variables grouped together in $\mathbf{y_2}(t,\mathbf{y_1})$ (solid).}\label{fig:DePaola_Structure}
\end{figure}
As can be seen, the coupling between the short- and long-term variables is one-way: the long-term dynamics depend on the short-term variables, but not vice versa.
For further insights into the biological interpretation of this model, we make reference to the original publications~\cite{de_paola_long-term_2023,de_paola_modeling_2025}
and the preceding works~\cite{roy_dynamic_2007,ha_mathematical_2016,morettini_system_2017}.

\subsection{Existence and uniqueness}\label{subsec:PicardLindeloef}

To prove existence and uniqueness of the solution we use standard arguments from the theory of ODEs.
The Picard-Lindelöf theorem, which requires the right-hand side $\mathbf{f}$ in Equation~\eqref{eq:y}
to be continuous in $t$ and Lipschitz-continuous in the state variables, 
guarantees local existence and uniqueness of a solution given the initial condition $\mathbf{y_0}$.

The Lipschitz-continuity with respect to the state variables can easily be verified.
In System~\eqref{eq:y1}, the right hand sides of the short-term equations are linear in each state variable, ensuring Lipschitz-continuity.
For System~\eqref{eq:y2}, the Lipschitz-continuity of the auxiliary functions has been established in Appendix~\ref{appendix:scaling}, which directly implies the Lipschitz-continuity of $\mathbf{f_2}(t,\mathbf{y_1},\mathbf{y_2})$.

The control $u(t)$, however, is discontinuous at time points
where the Heaviside function~\eqref{eq:control} changes values, whenever 
$t=k\nu + \delta$ or
$t=(k+1)\nu$ for $k = 0, \dots, n-1$.
Despite this, in the initial interval $t \in [0,\delta]$ all the conditions for the Picard-Lindelöf theorem are satisfied and there exists a unique solution given the initial condition $\mathbf{y_0}$.
For subsequent intervals, such as $t\in (\delta,\nu)$, we restart the system using the final value for $\mathbf{y}$ from the previous interval as initial condition,
again yielding existence and uniqueness in the interval.
By iteratively restarting the system, we extend existence and uniqueness to the entire interval $t\in[0,t_\mathrm{end}]$ for a given initial condition $\mathbf{y_0}$.
\section{Model reduction for computational speedup}\label{sec:homogenization}

When numerically solving the model introduced in Section~\ref{sec:full}, the
short-term equations require the numerical solver to take small time steps, in the order of minutes, to capture the periodic physical activity dynamics.
When simulating several years of diabetes progression, this is computationally intensive, especially if many simulations are required.
Instead of passing on the short-term effects encapsulated in $\mathbf{y_1}$ to the long-term
effects described in $\mathbf{y_2}$, we aim to replace
these fast fluctuations with constant values that represent the average contribution of the
short-term effects to $\mathbf{y_2}$.
This reduction follows the idea of periodic averaging~\cite{sanders_averaging_2007}. In
this case, the averaging procedure can be carried out analytically due to the structure of the control function $u(t)$ and the form of the short-term dynamics.

\subsection{Average solution of the short-term equations}

In the interval $[0,\delta]$, where $u(t) = 1$ and $\mathbf{y_1}(0) = \mathbf{0}$, the analytical solution for System~\eqref{eq:y1} is given by:
\[
\mathbf{y_1}(t)
= 
\begin{bmatrix}
\VO(t) \\
\Gpr(t) \\
\Gup(t) \\
\Ie(t) \\
\IL(t)
\end{bmatrix}
=
\begin{bmatrix}
1 - \exp(-\lambda_t \theta t) \\
1 + c_1(\alpha_2)\exp(-\lambda_t \theta t) - c_2(\alpha_2) \exp(-\lambda_t \alpha_2 t) \\
1 + c_1(\alpha_4)\exp(-\lambda_t \theta t) - c_2(\alpha_4) \exp(-\lambda_t \alpha_4 t) \\
1 + c_1(\alpha_6)\exp(-\lambda_t \theta t) - c_2(\alpha_6) \exp(-\lambda_t \alpha_6 t) \\
1 + c_1(\kappa_{\mathrm{IL6}})\exp(-\lambda_t \theta t) - c_2(\kappa_{\mathrm{IL6}}) \exp(-\lambda_t \kappa_{\mathrm{IL6}} t)
\end{bmatrix},
\]
where we define the two functions
\[
c_1(\pi) \isdef \frac{\pi}{\theta - \pi} \quad \text{and}\quad
c_2(\pi) \isdef \frac{\theta}{\theta - \pi}
\]
in order to write the parameters more compactly.

In the interval $(\delta,\nu)$, where $u(t) = 0$ and the initial condition is given by 
evaluating the previous solution at $t = \delta$, the analytical solution for the System~\eqref{eq:y1} is given by:
\[
\mathbf{y_1}(t)
= 
\begin{bmatrix}
\VO(t) \\
\Gpr(t) \\
\Gup(t) \\
\Ie(t) \\
\IL(t)
\end{bmatrix}
=
\begin{bmatrix}
\big(\exp(\lambda_t \theta \delta) -1\big)\exp(-\lambda_t \theta t) \\
c_4(\alpha_2) \exp(-\lambda_t \alpha_2 t) -c_3(\alpha_2)\exp(-\lambda_t \theta t)\\
c_4(\alpha_4) \exp(-\lambda_t \alpha_4 t)-c_3(\alpha_4)\exp(-\lambda_t \theta t) \\
c_4(\alpha_6) \exp(-\lambda_t \alpha_6 t)-c_3(\alpha_6)\exp(-\lambda_t \theta t) \\
c_4(\kappa_{\mathrm{IL6}}) \exp(-\lambda_t \kappa_{\mathrm{IL6}} t)-c_3(\kappa_{\mathrm{IL6}})\exp(-\lambda_t \theta t)
\end{bmatrix},
\]
where we
introduce
\[
c_3(\pi) \isdef \big(\exp(\lambda_t \theta \delta) -1\big)c_1(\pi), \quad \text{and}\quad
c_4(\pi) \isdef \big(\exp(\lambda_t \pi \delta )-1\big)c_2(\pi),
\]
again to write the parameters in compact form.
The steps leading up to these solutions are given in Appendix~\ref{appendix:reduction}.

In the next interval, $[\nu,\nu+\delta]$, we proceed analogously, noting that $\mathbf{y_1}(\nu)$ is close to zero. Note that the four parameters $\alpha_2$, $\alpha_4$, $\alpha_6$ and $\kappa_{\mathrm{IL6}}$ are much smaller than $\theta$ (see Appendix~\ref{appendix:parameters}), such that $0\leq \mathbf{y_1}(t)\leq 1$ holds for $t\in[0,T]$.
This procedure could be continued in order to construct a
semi-analytical solution of the short-term equations.
However, the solver would still need to take small
time steps to resolve the dynamics when using the semi-analytical solution.

Hence, we calculate the average $\mathbf{\mu}$ of $\mathbf{y_1}(t)$ in the interval $[0,\nu)$, defined as
\[
\mathbf{\mu}
\isdef
\frac{1}{\nu} \int_0^\nu 
\mathbf{y_1}(t)
\,\mathrm{d}t
=
\frac{1}{\nu}\bigg(\int_0^\delta \mathbf{y_1}(t) \,\mathrm{d}t + \int_\delta^\nu \mathbf{y_1}(t) \,\mathrm{d}t\bigg).
\]
A calculation, with details in Appendix~\ref{appendix:reduction}, yields
\[
\mathbf{\mu} = 
\begin{bmatrix}
\mu_{\VO} \\
\mu_{\Gpr} \\
\mu_{\Gup} \\
\mu_{\Ie} \\
\mu_{\IL}
\end{bmatrix}
=
\frac{1}{\nu}
\begin{bmatrix}
\delta
-\varepsilon \\
\delta + c_5(\alpha_2) - c_6(\alpha_2) \\
\delta + c_5(\alpha_4) - c_6(\alpha_4) \\
\delta + c_5(\alpha_6) - c_6(\alpha_6) \\
\delta + c_5(\kappa_{\mathrm{IL6}}) - c_6(\kappa_{\mathrm{IL6}})
\end{bmatrix},
\]
where we introduce
\[
\varepsilon \isdef \frac{\exp(-\lambda_t \theta\nu)}{\lambda_t \theta}\big(\exp(\lambda_t \theta \delta) -1\big),
\]
and
\[
c_5(\pi) \isdef \frac{\exp(-\lambda_t \theta\nu)}{\lambda_t \theta}c_3(\pi),
\quad
c_6(\pi) \isdef \frac{\exp(-\lambda_t\pi\nu)}{\lambda_t\pi}c_4(\pi),
\]
building on the previous results.

\subsection{Reduced model}

We use $\mathbf{\mu}$ to approximate the oscillating behavior of the short-term effects described by $\mathbf{y_1}$ and introduce an approximation $\widehat{\mathbf{y}}$ to System~\eqref{eq:y}:
\begin{equation}\label{eq:yhat}
\dt \widehat{\mathbf{y}}(t)
= \dt
\begin{bmatrix}
\widehat{\mathbf{y_1}}(t) \\
\widehat{\mathbf{y_2}}(t,\widehat{\mathbf{y_1}})
\end{bmatrix}
=
\begin{bmatrix}
\widehat{\mathbf{f_1}}(t, \widehat{\mathbf{y_1}}) \\
\mathbf{f_2}(t, \widehat{\mathbf{y_1}}, \widehat{\mathbf{y_2}})
\end{bmatrix}
= \widehat{\mathbf{f}}(t, \widehat{\mathbf{y}}),
\end{equation}
for $0< t \leq t_\mathrm{end}$, together with an initial condition defined below. 
For $\widehat{\mathbf{y_1}}$, we define a mock system
\[
\dt \widehat{\mathbf{y_1}}(t)
= \dt
\begin{bmatrix}
\widehat{\VO} \\
\widehat{\Gpr} \\
\widehat{\Gup} \\
\widehat{\Ie} \\
\widehat{\IL}
\end{bmatrix}
=
\begin{bmatrix}
0 \\
0 \\
0 \\
0 \\
0 \\
\end{bmatrix}
= \widehat{\mathbf{f_1}}(t,\widehat{\mathbf{y_1}}),
\]
with the initial conditions
$\widehat{\mathbf{y_1}}(0) \isdef \mathbf{\mu}$.
The right hand side for $\widehat{\mathbf{y_2}}$ remains as written out in System~\eqref{eq:y2}, but now gets the inputs from $\widehat{\mathbf{y_1}}$:
\begin{equation}\label{eq:y2hat}
\dt \widehat{\mathbf{y_2}}
= \dt
\begin{bmatrix}
\widehat{\VL} \\
\widehat{\SI} \\
\widehat{\Gamma} \\
\widehat{\Sigma} \\
\widehat{B} \\
\widehat{I} \\
\widehat{G}
\end{bmatrix}
=
\begin{bmatrix}
h_{\VL}(\widehat{\VL},\widehat{\IL})\\
h_{\SI}(\widehat{\SI},\widehat{\VL})\\
h_{\Gamma}(\widehat{\Gamma},\widehat{G})\\
h_{\Sigma}(\widehat{\Sigma},\widehat{\Gamma},\widehat{G})\\
h_{B}(\widehat{B},\widehat{\VL},\widehat{\Gamma},\widehat{\Sigma},\widehat{G})\\
h_{I}(\widehat{I},\widehat{\Ie},\widehat{\Gamma},\widehat{\Sigma},\widehat{B},\widehat{G})\\
h_{G}(\widehat{G},\widehat{\Gup},\widehat{\Gpr},\widehat{\SI},\widehat{I})
\end{bmatrix}
= \mathbf{f_2}(t,\widehat{\mathbf{y_1}},\widehat{\mathbf{y_2}}),
\end{equation}
with the previously given initial conditions.
Note that $\widehat{\mathbf{y}}(t)$ no longer directly depends on $u(t)$, because the control only influences the system via its parameters.
The existence and uniqueness of solutions to System~\eqref{eq:yhat} directly follows from the existence and uniqueness shown in Subsection~\ref{subsec:PicardLindeloef}.
From a numerical perspective, it is not necessary to solve System~\eqref{eq:yhat} explicitly. 
Instead, we can directly solve System~\eqref{eq:y2hat}, using the fact that
$\widehat{\mathbf{y_1}}(t) = \mathbf{\mu}$ for $0\leq t \leq t_\mathrm{end}$.

In classical homogenization and averaging theory, the separation of scales is
typically introduced through a small parameter $\varepsilon$ such that fast
dynamics evolve on the scaled time $t/\varepsilon$.
In our formulation, the constant parameter $\lambda_t$ plays a similar
role to $1/\varepsilon$. To stay consistent with the physiological
interpretation of the model, we keep the notation $\lambda_t$ and do
not study the asymptotic limit $\lambda_t \to \infty$.
Instead, we focus on a practical reduction of the system for
fixed values of $\lambda_t$, relevant to the long-term effects
of regular physical activity.

\subsection{Approximation error}\label{subsec:err}

We show that the error introduced by the homogenization procedure is bounded over a finite time horizon, for a fixed separation of timescales determined by the parameter $\lambda_t$.
Rather than analyzing the asymptotic limit, we quantify the approximation 
at physiologically relevant parameter values. To this end, we define the vector-valued $L^\infty$ norm as follows:
\[
\|\mathbf{g}\|_{L^\infty([0,t_\mathrm{end}])} = \max_{i=1,\dots,m} \bigg\{\mathop{\operatorname{ess\,sup}}_{t \in [0,t_\mathrm{end}]} |g_i(t)| \bigg\}
\]
for a function $\mathbf{g}\in L^\infty([0,t_\mathrm{end}],\mathbb{R}^m)$.
This error reflects the largest deviation
across all state variables over the time horizon.
\begin{theorem}\label{thm:1}
Let $\mathbf{y}(t,u)$ as introduced in System~\eqref{eq:y} and $\widehat{\mathbf{y}}(t)$ as introduced in System~\eqref{eq:yhat}.
Then there exists a constant $C$ depending on $t_\mathrm{end}$ such that
\[
\|\mathbf{y}(t,u) - \widehat{\mathbf{y}}(t)\|_{L^\infty([0,t_\mathrm{end}])} \leq C(t_\mathrm{end}).
\]
\end{theorem}
\begin{proof}
We can split the error as follows:
\begin{align*}
\|\mathbf{y}(t,u) - \widehat{\mathbf{y}}(t)\|_{L^\infty([0,t_\mathrm{end}])} = \max\big\{
&\|\mathbf{y_1}(t,u) - \widehat{\mathbf{y_1}}(t)\|_{L^\infty([0,t_\mathrm{end}])},\\
&\|\mathbf{y_2}(t,\mathbf{y_1}) - \widehat{\mathbf{y_2}}(t,\widehat{\mathbf{y_1}})\|_{L^\infty([0,t_\mathrm{end}])} \big\}.
\end{align*}
Since $\widehat{\mathbf{y_1}}(t) = \mathbf{\mu}$ by construction, we can bound the first part directly:
\[
\|\mathbf{y_1}(t,u) - \widehat{\mathbf{y_1}}(t)\|_{L^\infty([0,t_\mathrm{end}])} =
\|\mathbf{y_1}(t,u) - \mathbf{\mu}\|_{L^\infty([0,t_\mathrm{end}])} =
\|1 - \mathbf{\mu}\|_{L^\infty},
\]
where we used that $\mathbf{y_1}(t,u) \in [0,1]$ due to the scaling and that every element of $\mathbf{\mu}$ is smaller than $0.5$.
In the $L^\infty$-norm, this error is independent of $t_\mathrm{end}$ by construction, but larger than it would be in $L^2$-norm for example.

For the second part, we first note that
\begin{alignat*}{2}
\|\mathbf{f_2}(t,\mathbf{y_1},\mathbf{y_2}) - \mathbf{f_2}(t,\widehat{\mathbf{y_1}},\widehat{\mathbf{y_2}})\|_{L^\infty([0,t_\mathrm{end}])}
&\leq &&\mathbf{L_1}\|\mathbf{y_1} - \widehat{\mathbf{y_1}}\|_{L^\infty([0,t_\mathrm{end}])}\\
& && + \mathbf{L_2}\|\mathbf{y_2} - \widehat{\mathbf{y_2}}\|_{L^\infty([0,t_\mathrm{end}])}\\
&\leq &&\mathbf{L_1}\|1 - \mathbf{\mu}\|_{L^\infty}\\
& &&+ \mathbf{L_2}\|\mathbf{y_2} - \widehat{\mathbf{y_2}}\|_{L^\infty([0,t_\mathrm{end}])}.
\end{alignat*}
with Lipschitz constants $\mathbf{L_1}$ and $\mathbf{L_2}$. This holds by definition of the Lipschitz-continuity of $\mathbf{f_2}$.
We now invoke the Grönwall-Lemma in its differential form to bound
$\|\mathbf{y_2}(t,\mathbf{y_1}) - \widehat{\mathbf{y_2}}(t,\widehat{\mathbf{y_1}})\|_{L^\infty([0,t_\mathrm{end}])}$
and get that
\[
\|\mathbf{y_2}(t,\mathbf{y_1}) - \widehat{\mathbf{y_2}}(t,\widehat{\mathbf{y_1}})\|_{L^\infty([0,t_\mathrm{end}])}
\leq \mathbf{L_1}\|1 - \mathbf{\mu}\|_{L^\infty}
\int_{0}^t \exp\big(\mathbf{L_2}(t-s)\big)  \, \mathrm{d}s,
\]
for $t \in [0,t_\mathrm{end}]$. A calculation yields
\begin{align*}
\|\mathbf{y_2}(t,\mathbf{y_1}) - \widehat{\mathbf{y_2}}(t,\widehat{\mathbf{y_1}})\|_{L^\infty([0,t_\mathrm{end}])} &\leq \mathbf{L_1}\|1 - \mathbf{\mu}\|_{L^\infty}\exp(\mathbf{L_2} t)
\int_{0}^t\exp(-\mathbf{L_2}s)  \, \mathrm{d}s\\
&= \frac{\mathbf{L_1}}{\mathbf{L_2}}\|1 - \mathbf{\mu}\|_{L^\infty}
\big(\exp(\mathbf{L_2} t)-1\big)
\end{align*}
for $t \in [0,t_\mathrm{end}]$.
It follows that
\[
\|\mathbf{y}(t,u) - \widehat{\mathbf{y}}(t)\|_{L^\infty([0,t_\mathrm{end}])} \leq C(t_\mathrm{end}).
\]
\end{proof}
A related proof can be found in Sanders~et~al.~\cite[Section 2.8]{sanders_averaging_2007}, which presents a classical averaging result in the asymptotic regime $\varepsilon \to 0$. In contrast, our setting assumes a fixed separation of timescales via the parameter $\lambda_t \sim 1/\varepsilon$. Rather than analyzing the asymptotic limit, we aim to establish a practical error bound for physiologically relevant values of $\lambda_t$ over a finite time horizon. The bound in Theorem~\ref{thm:1} therefore depends on $t_\mathrm{end}$ and reflects the cumulative effect of replacing the periodic short-term dynamics with their average.
\section{Numerical results}\label{sec:numerics}

The reduced model provides a significant computational speedup compared to
the full model, as it does not require resolving minute-scale dynamics.
This allows the numerical solver to take time steps on the order of days rather than minutes. Furthermore, the numerical results underline that the two systems behave very similarly, as presented in the following.

\subsection{Illustration}

System~\eqref{eq:y} with the standard parameters listed in Appendix~\ref{appendix:parameters} is stiff. To solve it, 
we use an implicit solver and constrain the step size to be small enough
to catch the control dynamics.
In our experiments, a maximum step size of 1 hour with a fifth-order
implicit Runge-Kutta method was used.
An illustration of the numerical solutions for one month of the full and the reduced model can be found in Figure~\ref{fig:Sol_1month}.
\begin{figure}[htp]
    \centering
\includegraphics[width=0.8\textwidth]{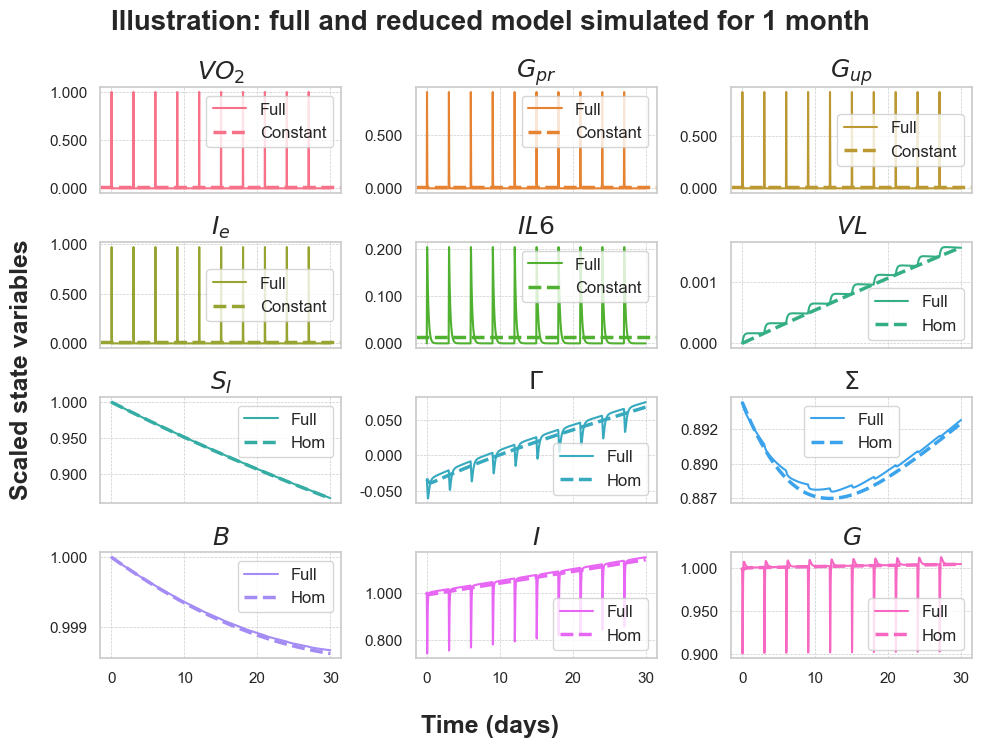}
    \caption{Illustration of the numerical solution of the full model (solid lines) and the reduced model (dashed lines) for all 12 state variables for one month. All the parameters are chosen as outlined in Appendix~\ref{appendix:parameters}.}\label{fig:Sol_1month}
\end{figure}
To examine the long-term behavior of both the full and the reduced model, we focus on the solutions at basal times, referring to the moments of the start of physical activity.
Figure~\ref{fig:Sol_5y_basal}
presents the solutions over a five-year period.
\begin{figure}[htp]
    \centering
    \includegraphics[width=0.8\textwidth]{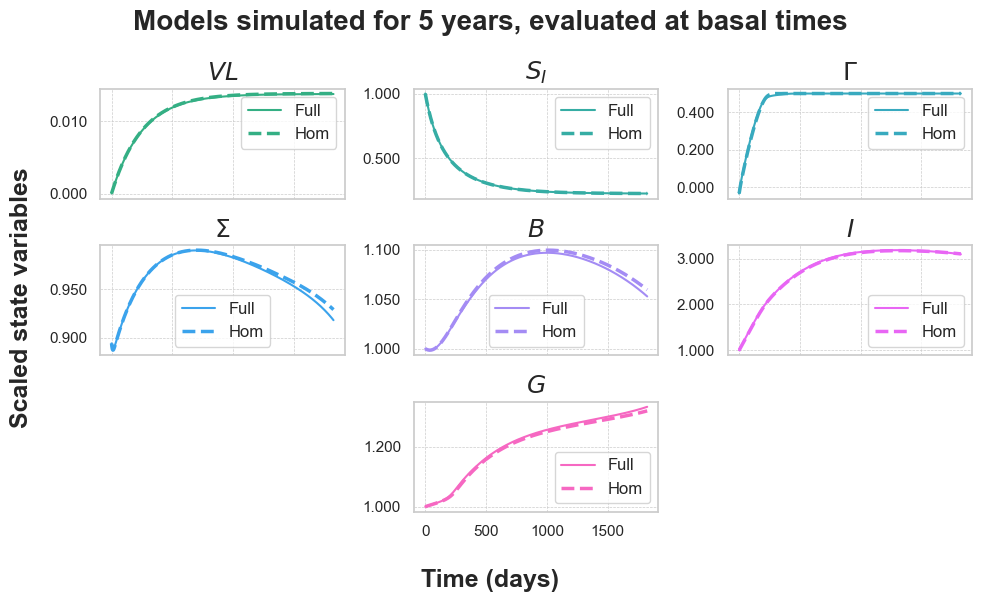}
    \caption{Solution at basal times (moments of the start of physical activity) of the full model (solid lines) and the reduced model (dashed) for five years. All the parameters are chosen as outlined in Appendix~\ref{appendix:parameters}.}\label{fig:Sol_5y_basal}
\end{figure}

\subsection{Simulation study}

To assess the speedup and the approximation error made by reducing the system,
we vary six parameters and three initial conditions, resulting
in $19\,683$ parameter configurations, and simulate the full and the reduced model over a 5-year period.
Following the definition of the approximation error in Subsection~\ref{subsec:err},
we define the maximal deviation at basal times between the two systems over 5 years for simulation $i$ as $E_i^{\mathrm{basal}}$. This metric captures the largest discrepancy between the two systems
and thus provides a conservative upper bound on the approximation error.
We also consider the pointwise error at 5 years, $E_i^{\mathrm{5y}}$.
Since classification into normoglycemia or
type 2 diabetes is based on glucose thresholds, we use $E_i^{\mathrm{5y}}$ to check whether the
reduced model does
not incorrectly classify subjects as normoglycemic when the full model would indicate type 2 diabetes.
Details on the parameters and the error measures
can be found in Appendix~\ref{appendix:simulations}.
All simulations were performed using the
\texttt{solve\_ivp} implementation with \texttt{method = 'Radau'} from the \texttt{scipy} package
on a computing server equipped with a 64 core AMD EPYC 7763 CPU and one Nvidia GeForce RTX 4090 GPU.

On average, simulating the full model over a five-year period requires 98 seconds (SD 26 seconds).
The reduced model requires no constraint on the step size,
resulting in a substantial computational speedup.
On average, one simulation over the five-year period takes 0.053 seconds (SD 0.01 seconds), corresponding to a reduction in
computational time of a factor 1914 (SD 655).
This slighlty exceeds the theoretically expected speedup factor of $\lambda_t$ = 1440, reflecting the relationship between the two timescales.

A box plot summarizing the distribution of $E_i^{\mathrm{basal}}$ across the $19\,683$ simulations is shown in Figure~\ref{fig:Linfty}.
\begin{figure}[htp]
    \centering
    \includegraphics[width=0.8\textwidth]{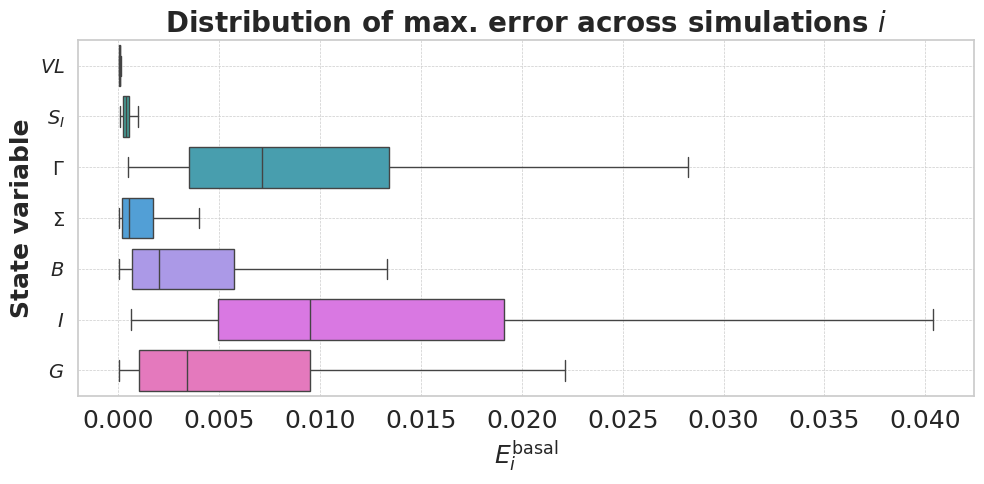}
    \caption{Box plot showing the distribution of the maximal deviation $E_i^{\mathrm{basal}}$ between the full and the reduced model across $19\,683$ simulations.}\label{fig:Linfty}
\end{figure}
These findings underline that the maximal error across the simulations remains reasonably small.
A box plot summarizing the distribution of $E_i^{\mathrm{5y}}$ across the $19\,683$ simulations is given in Figure~\ref{fig:5y}.
\begin{figure}[htp]
    \centering
    \includegraphics[width=0.8\textwidth]{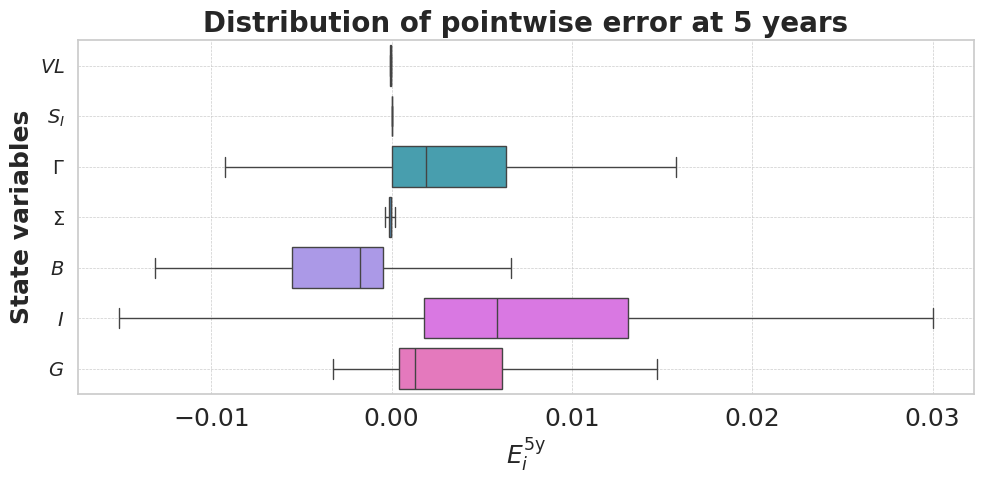}
    \caption{
    Box plot showing the distribution of the deviation $E_i^{\mathrm{5y}}$ between the full and the reduced model across $19\,683$ simulations at 5 years.}\label{fig:5y}
\end{figure}
The pointwise error of glucose $G$ at 5 years is small and negative in only $5$ out of $19\,683$ simulations.
This confirms that the reduced model closely approximates the full model while achieving a $\lambda_t$-fold speedup. Furthermore, the risk of misclassification based on glucose values is minimal.
\section{Discussion}\label{sec:discussion}

In this work, we developed and analyzed a reduced,
computationally efficient version
of a two-scale model capturing the
long-term effects of physical activity on blood
glucose regulation.
The reduced model retains the essential dynamics of
the original formulation while decreasing the
computational cost of simulating multi-year trajectories to type 2 diabetes.
This gain in efficiency allows for extensive patient-specific simulations and
the systematic evaluation of
physical activity plans in a personalized context.

The computational speedup was achieved by analytically
reducing the original model using
the periodic nature of the physical activity input (control variable $u$).
We analytically solved the short-term effects of physical activity
and passed them to the long-term
dynamics as averaged quantities.
By treating the reduction as a perturbation of the
right-hand side, we proved that the approximation
error remains bounded.
These results are underlined with a simulation study where
nine key parameters were varied.
This confirmed that the
reduced model reproduces the full model’s dynamics with
negligible long-term discrepancies.
The resulting model lowers the simulation time of a
five-year trajectory
by almost a factor 2000.
Indeed, simulating a single parameter set over a five-year period takes around 0.05 seconds,
making the fast simulation of many scenarios feasible.

The proposed formulation was designed to solve
the original system with minimal
computational effort while maintaining accuracy.
Alternative simplifications, such as using semi-analytical solutions
of the short-term equations directly
or replacing the physical activity input with a constant average input,
would still require
the numerical solver to take small time steps to resolve
the oscillations. In contrast,
the proposed formulation removes these fast variations analytically,
allowing the solver to take large time steps (i.e., days instead of minutes) and achieve a significant
computational speedup, while retaining the physiologically meaningful link between exercise and glucose regulation.

The considerable reduction in computational effort extends
the applicability of the reduced model to the field of medical decision support.
Personalized physical activity plans, along with the
uncertainty around their impact, can be evaluated
efficiently, for example in a Bayesian framework.
This model has been integrated into a causal learning framework
to assess the impact of
hypothetical physical activity plans in slowing down or preventing progression to type 2 diabetes in at-risk individuals~\cite{lenatti_counterfactual_2025}.
This approach provides a foundation for what-if analyses and supports personalized decision-making in clinical settings.

On the theoretical side, we established formal existence
and uniqueness results and introduced a scaled formulation
that improves numerical stability.
The reduction approach draws on concepts from homogenization theory
but departs from its classical asymptotic framework,
which typically examines model behavior as the timescale
separation tends to
infinity~\cite{sanders_averaging_2007,pavliotis_multiscale_2008}.
In contrast, the proposed formulation assumes a fixed, physiologically
meaningful separation of timescales and evaluates how
accurately the reduced model reproduces the full system
over a realistic five-year horizon.
A more formal asymptotic analysis based on averaging
theory and including convergence
results would be a valuable direction for future work.

As with previous formulations~\cite{de_paola_long-term_2023,de_paola_modeling_2025}, this model describes
an average patient and assumes that physical activity
follows a periodic pattern over long timescales. 
It does not explicitly account for factors such as age,
sex, family history of type 2 diabetes, or diet.
Furthermore, it has not been validated against longitudinal clinical data due to the scarcity of such datasets.
The nonlinear interaction terms driving the glucose-insulin
dynamics also contribute to system stiffness and
increase the complexity of both analysis and numerical
simulation~\cite{li_modeling_2006,shabestari_new_2018}.

In summary, this work introduces a mathematically grounded
and computationally efficient formulation of a physiologically
detailed glucose regulation model.
By leveraging homogenization principles,
we achieve a reduction that enables fast,
accurate long-term simulations and extends the
model’s usability to personalized type 2 diabetes risk prediction.
More broadly, this approach illustrates how mechanistic
models in systems biology can benefit from mathematical
reduction and analysis techniques to achieve both
rigor and scalability.

\section{Acknowledgment}
This work was supported by the European Union and by the Swiss State Secretariat for Education, Research and Innovation (SERI) through project PRAESIIDIUM
”Physics informed machine learning-based prediction and reversion of impaired
fasting glucose management” under Grant 101095672.
Views and opinions expressed are however those of the authors
only and do not necessarily reflect those of the European Union and SERI.
The European Union and SERI cannot be held responsible for them.

Pierluigi Francesco De Paola is a PhD student enrolled in the National PhD Program in Autonomous Systems (DAUSY), coordinated by Politecnico of Bari, Bari, Italy.

Marta Lenatti is a PhD
student enrolled in the National PhD in Artificial Intelligence, XXXVIII cycle, course
on Health and life sciences, organized by Università Campus Bio-Medico di Roma.

The authors thank Fabrizio Dabbene, Alessandro Borri and
Pasquale Palumbo for their valuable discussions about the full model.

\section{Code availability}
The code supporting this work is publicly available at:\\
\url{https://gitlab-core.supsi.ch/dti-idsia/Homogen_ODEs_T2D}

\section{Appendix}

\subsection{Scaling of the model}\label{appendix:scaling}

We adopt the following notation:
state variables are represented by capital letters,
using either Greek or Latin characters.
Parameters, which are always constant in time, are consistently denoted by lowercase
Greek letters and may include subscripts.
Functions are written as lowercase Latin letters, and also may include subscripts.

\subsubsection{Scaling of the short-term equations}

We rescale the original system~\cite{de_paola_long-term_2023, de_paola_modeling_2025} based on the magnitude of the state variables and introduce a unified timescale.
The equations on a minute-scale read
\begin{alignat*}{2}
& \dt \VO &&= \theta u(t) - \theta \VO, \\
& \dt \Gpr &&= \alpha_1 \VO - \alpha_2 \Gpr, \\
& \dt \Gup &&= \alpha_3 \VO - \alpha_4 \Gup, \\
& \dt \Ie &&= \alpha_5 \VO - \alpha_6 \Ie, \\
& \dt \IL &&= \kappa_\mathrm{SR} \VO - \kappa_{\mathrm{IL6}} \IL, \\
& \dt \VL &&= \IL - \kappa_{\mathrm{s}} \VL,
\end{alignat*}
with the initial conditions being $0$ for all the equations.
The control $u(t)$ is introduced as
\[u(t) = 
\begin{cases}
\xi & \text{for } 0 \leq t \leq \text{duration (min)},\\
0 & \text{for } \text{duration (min)} < t < \text{period length (min)},
\end{cases}
\]
where $\xi$ is the physical activity intensity. The period length refers to the time between the start of two consecutive physical activity sessions.
Note that the state variable $\VL$ is included in this set of equations because it was originally parameterized in minutes
and hence requires the short-term scaling.
However, since it has a long-term effect we will
include it in the long-term equations for the homogenization further below.

The original units of the state variables, along with their description, are given in Table~\ref{table:units_short_term}.
\begin{table}[!htb]
    \centering
    \caption{Short description and units of the short-term state variables in their unscaled form.}\label{table:units_short_term}
    \begin{tabular}{r l l}
    \toprule
Variable & Description & Unit \\
        \midrule
        $\VO$ & Oxygen consumption during exercise & given in \%\\
        $\Gpr$ & Incremental hepatic glucose production & mg/(kg min) \\
        $\Gup$ & Increased glucose uptake by working tissues & mg/(kg min) \\
        $\Ie$ & Incremental insulin removal & $\mu$U/ml \\
        $\IL$ & Concentration of IL-6 in the muscle & pg/ml\\
        $\VL$ & Integral effect of IL-6 released during exercise & (pg/ml) min \\
        \bottomrule
    \end{tabular}
\end{table}
To solve these ODEs simultaneously with the long-term equations defined at a daily timescale,
we first convert the short-term equations to days.
Additionally, we scale the state variables to be dimensionless.
Following standard procedures for the scaling of
ODEs~\cite{langtangen_scaling_2016}, we introduce
\[
\overline{t} \isdef \frac{t}{\lambda_t}\;\text{and}\;
\overline{\Omega} \isdef \frac{\Omega}{\lambda_\Omega}, \; \Omega \in \{t, \VO, \Gpr, \Gup, \Ie, \IL, \VL \}.
\]
Substituting into both sides of the unscaled system and rearranging yields
\begin{alignat*}{2}
& \dtbar \overline{\VO} &&= \lambda_t\bigg(\frac{\theta}{\lambda_{\VO}} u(\lambda_t \overline{t}) - \theta \overline{\VO}\bigg), \\
& \dtbar \overline{\Gpr} &&=
\lambda_t\bigg(\frac{\alpha_1 \lambda_{\VO}}{\lambda_{\Gpr}}\overline{\VO} - \alpha_2 \overline{\Gpr}\bigg), \\
& \dtbar \overline{\Gup} &&=
\lambda_t\bigg(\frac{\alpha_3 \lambda_{\VO}}{\lambda_{\Gup}}\overline{\VO} - \alpha_4 \overline{\Gup}\bigg), \\
& \dtbar \overline{\Ie} &&=
\lambda_t\bigg(\frac{\alpha_5 \lambda_{\VO}}{\lambda_{\Ie}}\overline{\VO} - \alpha_6 \overline{\Ie}\bigg), \\
& \dtbar \overline{\IL} &&=
\lambda_t\bigg(\frac{\kappa_\mathrm{SR}\lambda_{\VO}}{\lambda_{\IL}}\overline{\VO} - \kappa_\mathrm{IL6} \overline{\IL}\bigg), \\
& \dtbar \overline{\VL} &&=
\lambda_t\bigg(\frac{\lambda_{\IL}}{\lambda_{\VL}}\overline{\IL} - \kappa_{\mathrm{s}} \overline{\VL}\bigg).
\end{alignat*}
The initial conditions are $0$ for all equations, hence no scaling is necessary.
We now select the scaling constants as specified in Appendix~\ref{appendix:parameters}. Note that the choice of $\lambda_t$ transforms the time from minutes to days. All the other constants are chosen to make the state variables unitless and to scale the system to take values in the interval $[0,1]$.
We conclude the scaling of the short-term state variables by noting that, to include them in the long-term equations, they must be rescaled to their original units. In particular, state variables with minute-based units must be converted to days as follows:
\[
\Gpr^{\textrm{day}} \isdef \lambda_t \Gpr, \;
\Gup^{\textrm{day}} \isdef \lambda_t \Gup, \;
\VL^{\textrm{day}} \isdef \frac{\VL}{\lambda_t}.
\]

\subsubsection{Auxiliary functions for the long-term equations}

We provide a concise overview of the set of auxiliary functions originally introduced in the appendix of the publication
by Ha et~al.~\cite{ha_mathematical_2016}.
We first define the two functions $q_h$ and $q_e$ as follows:
\begin{align*}
q_h(x;\pi_1,\pi_2) &\isdef \frac{x^{\pi_2}}{x^{\pi_2} + \pi_1^{\pi_2}},
\; \text{for}\; \pi_1\in\mathbb{R}^+, \; \pi_2 \in 2\mathbb{N}^+,\\
q_e(x;\pi_1,\pi_2,\pi_3,\pi_4) &\isdef \frac{\pi_1}{1 + \pi_4 \exp\big(-\frac{x - \pi_2}{\pi_3}\big)} \; \text{for}\;\pi_1, \pi_2, \pi_3, \pi_4\in\mathbb{R}^+,
\end{align*}
where $x\in\mathbb{R}$.
Under the listed constraints on the constants, both functions are Lipschitz-continuous and bounded in $x$.
With these functions at hand, we define the following
auxiliary functions, following the notation of the original publications wherever possible.
A list of all the parameters introduced here can be found in Appendix~\ref{appendix:parameters}.
The first set of functions is based on the Hill-type function $q_h$:
\begin{align*}
m(G) &= q_h(G;\alpha_\mathrm{M},2),\\
r(\Gamma, \Sigma, G) &= \Sigma q_h\big(m(G) + \Gamma;\alpha_\mathrm{ISR},2\big),\\
p(\VL, \Gamma, \Sigma, G) &= \phi_{\mathrm{max}}
q_h\big(r(\Gamma, \Sigma, G);\alpha_\mathrm{P},4\big)\big(1 + \zeta_1q_h(\VL;\kappa_\mathrm{n},2)\big),\\
a(\VL, G) &= \Big(
\alpha_\mathrm{max}q_h\big(m(G);\alpha_\mathrm{A},6\big) + \alpha_\mathrm{B} \Big) \Big(1 - \zeta_2q_h(\VL;\kappa_\mathrm{n},2)\Big).
\end{align*}
It can be shown that all of these functions are Lipschitz-continuous with respect to all of their arguments.
Furthermore, it can be verified that the functions $m$, $p$ and $a$ are bounded. The function $r$ is only bounded if $\Sigma$ is bounded.

The second set requires the function $q_e$:
\begin{align*}
g_\infty(G) &= q_e(G;\gamma_{\mathrm{max}},\gamma_\mathrm{S},\gamma_\mathrm{n},1) - \gamma_\theta,\\
s_{\mathrm{ISR\infty}}(\Gamma, \Sigma, G) &= 
q_e\big(r(\Gamma, \Sigma, G - \kappa_{\sigma\mathrm{s}});\sigma_{\mathrm{ISRmax}},\sigma_\mathrm{ISRs},\sigma_\mathrm{ISRn},\sigma_{\mathrm{ISRk}}\big),\\
s_{\mathrm{M\infty}}(G) &= 1 - q_e\big(m(G - \kappa_{\sigma\mathrm{s}});\sigma_{\mathrm{Mmax}},\sigma_\mathrm{Ms},\sigma_\mathrm{Mn},\sigma_{\mathrm{Mk}}\big),\\
s_{\infty}(\Gamma, \Sigma, G) &= s_{\mathrm{ISR\infty}}(\Gamma, \Sigma, G) s_{\mathrm{M\infty}}(G) + \sigma_\mathrm{B}.
\end{align*}
Again, it is clear that all of these functions are Lipschitz-continuous and bounded in all of their arguments, as long as all of the parameters are positive.

\subsubsection{Scaling of the long-term equations}

The long-term equations are written down in the original publication as follows:
\begin{alignat*}{2}
& \dt \SI &&= \bigg(\frac{\theta_{\SI} - \SI}{\tau_{\SI}}\bigg)
\bigg(1 - \zeta_3 \frac{\VL^{\textrm{day}}}{k_{n\SI} + \VL^{\textrm{day}}}\bigg),\\
& \dt \Gamma &&= \frac{g_{\infty}(G) - \Gamma}{\tau_{\Gamma}},\\
& \dt \Sigma &&= \frac{s_{\infty}( \Gamma, \Sigma, G) - \Sigma}{\tau_{\Sigma}},\\
& \dt B &&= \frac{p(\VL^{\textrm{day}}, \Gamma, \Sigma, G) - a(\VL^{\textrm{day}}, G)}{\tau_{B}}B,\\
& \dt I &&= \frac{r(\Gamma, \Sigma, G)}{\upsilon}B  - \kappa I - \Ie,\\
& \dt G &&= \rho_0 + \frac{\omega}{\upsilon_g}(\Gpr^{\textrm{day}} - \Gup^{\textrm{day}}) - (\eta_0 + \SI I) G,
\end{alignat*}
with initial conditions $\SI(0) = {\SI}_0$,
$\Gamma(0) = \Gamma_0$,
$\Sigma(0) = \Sigma_0$, $B(0) = B_0$,
$I(0) = I_0$ and $G(0) = G_0$.
The five auxiliary functions
$g_\infty$, $s_\infty$, $p$, $a$ and $r$ have been defined above.
The original units of the introduced long-term state variables, along with their description, are given in Table~\ref{table:units_long_term}.
\begin{table}[htbp]
    \centering
    \caption{Short description and units of the long-term state variables in their unscaled form.}\label{table:units_long_term}
    \begin{tabular}{r l l}
    \toprule
    Variable & Description & Unit \\
        \midrule
        $\SI$ & Insulin sensitivity & ml/($\mu$U day)\\
        $\Gamma$ & Shift of the glucose dependence & - \\
        $\Sigma$ & Insulin secretion capacity& $\mu$U/ ($\mu$g day) \\
        $B$ & Beta cell mass & mg \\
        $I$ & Serum insulin concentration & $\mu$U/ml \\
        $G$ & Plasma glucose concentration& mg/dl \\
        \bottomrule
    \end{tabular}
\end{table}

Again, we introduce the following scaling of the state variables
\[
\overline{\Omega} \isdef \frac{\Omega}{\lambda_\Omega}, \; \Omega \in \{\SI, \Gamma, \Sigma, B, I, G \}.
\]
The chosen scaling parameters, along with all the parameters introduced, are listed in Appendix~\ref{appendix:parameters}.
Note that $t$ does not require scaling since we want to solve the final system in days.
For simplification, we introduce the following auxiliary functions:
\begin{equation}
\begin{aligned}
d_{\lambda}\big(\overline{\VL}\big) &\isdef \frac{1}{{\lambda_{\SI}}}\Bigg(1 - \zeta_3\bigg(\frac{(\lambda_{\VL}/\lambda_t)\overline{\VL}}{\kappa_{\mathrm{n}\SI} + (\lambda_{\VL}/\lambda_t)\overline{\VL}}\bigg)\Bigg),\\
g_{\lambda}\big(\overline{G}\big) &\isdef \frac{1}{\lambda_\Gamma}g_{\infty}\big(\lambda_G\overline{G}\big),\\
s_{\lambda}\big(\overline{\Gamma}, \overline{\Sigma}, \overline{G}\big) &\isdef \frac{1}{\lambda_{\Sigma}}s_{\infty}\big(\lambda_\Gamma \overline{\Gamma}, \lambda_{\Sigma}\overline{\Sigma}, \lambda_G\overline{G}\big)\\
p_\lambda\big(\overline{\VL}, \overline{\Gamma}, \overline{\Sigma}, \overline{G}\big) &\isdef p\big((\lambda_{\VL}/\lambda_t)\overline{\VL}, \lambda_\Gamma \overline{\Gamma}, \lambda_{\Sigma}\overline{\Sigma}, \lambda_G\overline{G}\big),\\
a_\lambda\big(\overline{\VL}, \overline{G}\big) &\isdef a\big((\lambda_{\VL}/\lambda_t)\overline{\VL}, \lambda_G\overline{G}\big),\\
r_{\lambda}\big(\overline{\Gamma}, \overline{\Sigma}, \overline{G}\big) &\isdef \frac{\lambda_B}{\lambda_I\upsilon}r\big(\lambda_\Gamma \overline{\Gamma}, \lambda_{\Sigma}\overline{\Sigma}, \lambda_G\overline{G}\big).
\end{aligned}\label{eq:final_aux_fun}
\end{equation}
We furthermore define the following constants:
\[
\lambda_{\Ie I} \isdef \frac{\lambda_{\Ie}}{\lambda_I}, \quad
\rho_{\lambda} \isdef \frac{\rho_0}{\lambda_G}, \quad
\lambda_{tG} \isdef \frac{\lambda_t}{\lambda_G\upsilon_g}, \quad
\lambda_{\SI I} \isdef \lambda_{\SI}\lambda_I.
\]
After substituting the scaling parameters and rewriting, we obtain the following system of ODEs:
\begin{alignat*}{2}
& \dt \overline{\SI} &&= d_{\lambda}\big(\overline{\VL}\big)\bigg(\frac{\theta_{\SI} - \lambda_{\SI}\overline{\SI}}{\tau_{\SI}}\bigg),\\
& \dt \overline{\Gamma} &&= \frac{g_{\lambda}\big(\overline{G}\big) - \overline{\Gamma}}{\tau_{\Gamma}},\\
& \dt \overline{\Sigma} &&= \frac{s_{\lambda}\big(\overline{\Gamma}, \overline{\Sigma}, \overline{G}\big) - \overline{\Sigma}}{\tau_{\Sigma}},\\
& \dt \overline{B} &&= 
\frac{p_\lambda\big(\overline{\VL}, \overline{\Gamma}, \overline{\Sigma}, \overline{G}\big) - a_\lambda\big(\overline{\VL}, \overline{G}\big)}{\tau_{B}}
\overline{B},\\
& \dt \overline{I} &&= r_{\lambda}\big(\overline{\Gamma}, \overline{\Sigma}, \overline{G}\big)\overline{B}  - \kappa \overline{I} - \lambda_{\Ie I}\overline{\Ie},\\
& \dt \overline{G} &&= \rho_{\lambda} + \lambda_{tG}\omega\big(\lambda_{\Gpr}\overline{\Gpr} - \lambda_{\Gup}\overline{\Gup}\big) - (\eta_0 + \lambda_{\SI I} \overline{\SI} \overline{I})\overline{G},
\end{alignat*}
together with the initial conditions
$\SI(0) = 1$,
$\Gamma(0) = \Gamma_0/\gamma_{\textrm{max}} \defis \Gamma_{0\lambda}$,
$\Sigma(0) = \Sigma_0/\sigma_{\mathrm{ISRmax}} \defis \Sigma_{0\lambda}$, $B(0) = 1$,
$I(0) = 1$ and $G(0) = 1$.

\subsection{Parameters}\label{appendix:parameters}

A list of the parameters introduced in the short-term equations ($\mathbf{y_1}$) can be found in Table~\ref{table:params_short_term}.
\begin{table}[!htb]
\centering
\caption{Parameters related to the control $u(t)$ (top), the system of ODEs for $\mathbf{y_1}$ (middle) and the scaling of $\mathbf{y_1}$ (bottom). Standard values or the formula and the respective units are given in columns 2 and 3.}\label{table:params_short_term}
\begin{tabular}{r r l}
\toprule
Parameter & Definition & Unit \\
\midrule
$\nu$ & 3 & day\\
$\delta$ & 60/1440 & day\\
$\xi$ & 50 & \%\\
\midrule
$\theta$ & 0.8 & 1/min \\
$\alpha_1$ & 0.00158 & mg/(kg min$^2$)\\
$\alpha_2$ & 0.056 & 1/min\\
$\alpha_3$ & 0.00195 & mg/(kg min$^2$)\\
$\alpha_4$ & 0.0485 & 1/min\\
$\alpha_5$ & 0.00125 & $\mu$U/(ml min)\\
$\alpha_6$ & 0.075 & 1/min\\
$\kappa_\mathrm{SR}$ & 0.045 & pg/(ml min)\\
$\kappa_\mathrm{IL6}$ & 0.004 & 1/min\\
\midrule
$\lambda_t$ & $1440$ & min/day\\
$\lambda_{\VO}$ & $\xi$ & given in \%\\
$\lambda_{\Gpr}$ & $\lambda_{\VO}\alpha_1/\alpha_2$ & mg/(kg min)\\
$\lambda_{\Gup}$ & $\lambda_{\VO}\alpha_3/\alpha_4$ & mg/(kg min)\\
$\lambda_{\Ie}$ & $\lambda_{\VO}\alpha_5/\alpha_6$ & $\mu$U/ml\\
$\lambda_{\IL}$ & $\lambda_{\VO}\kappa_\mathrm{SR}/\kappa_{\mathrm{IL6}}$ & pg/ml\\
\bottomrule
\end{tabular}
\end{table}
The parameters related to the state variable $\VL$ are reported below together with the long-term state variables.
The physical activity parameters $\nu$ and $\delta$ need to be chosen to allow for a maximum of 400 minutes of exercise per week, the intensity parameter $\xi$ needs to lie within $[0,92]\%$.

A list of the parameters introduced for the auxiliary functions can be found in Table~\ref{table:params_aux}.
\begin{table}[!htb]
\centering
\caption{Parameters used for the auxiliary functions building upon $q_h$ (top) and the auxiliary functions building upon $q_e$ (bottom). Standard values or the formula and the respective units are given in columns 2 and 3.}\label{table:params_aux}
\begin{tabular}{r r l}
\toprule
Parameter & Definition & Unit\\
\midrule
$\alpha_\mathrm{M}$ & 150 & mg/dl\\
$\alpha_\mathrm{ISR}$ & 1.2 & -\\
$\phi_{\mathrm{max}}$ & 4.55 & 1/day\\
$\alpha_\mathrm{P}$ & 41.77 & $\mu$U/($\mu$g day)\\
$\zeta_1$ & $10^{-4}$ & -\\
$\kappa_\mathrm{n}$ & $10^{6}/\lambda_t$ & (pg/ml) day\\
$\alpha_\mathrm{max}$ & 9 & 1/day\\
$\alpha_\mathrm{A}$ & 0.44 & -\\
$\alpha_\mathrm{B}$ & 0.8 & 1/day\\
$\zeta_2$ & $10^{-4}$ & -\\
\midrule
$\gamma_{\mathrm{max}}$ & 0.2 & -\\
$\gamma_\mathrm{S}$ & 99.9 & -\\
$\gamma_\mathrm{n}$ & 1& -\\
$\gamma_\theta$ & 0.1& -\\
$\kappa_{\sigma\mathrm{s}}$ & 75& mg/dl \\
$\sigma_{\mathrm{ISRmax}}$ & 600& $\mu$U/($\mu$g day)\\
$\sigma_\mathrm{ISRs}$ & 0.1& -\\
$\sigma_\mathrm{ISRn}$ & 0.1& -\\
$\sigma_{\mathrm{ISRk}}$ & 1& -\\
$\sigma_{\mathrm{Mmax}}$ & 1& -\\
$\sigma_\mathrm{Ms}$ & 0.2& -\\
$\sigma_\mathrm{Mn}$ & 0.02& -\\
$\sigma_{\mathrm{Mk}}$ & 0.2& -\\
$\sigma_\mathrm{B}$ & 3& $\mu$U/($\mu$g day)\\
\bottomrule
\end{tabular}
\end{table}
All of these parameters are constrained to be positive.

A list of the parameters introduced for the long-term equations can be found in Table~\ref{table:params_long_term}.
\begin{table}[!htb]
\centering
\caption{Parameters related to the system of ODEs for $\mathbf{y_2}$ (top), the scaling of $\mathbf{y_2}$  (middle) and the initial conditions for $\mathbf{y_2}$ (bottom).
Standard values or the formula and the respective units are given in columns 2 and 3.
}\label{table:params_long_term}
\begin{tabular}{r r l}
\toprule
Parameter & Definition & Unit \\
\midrule
$\kappa_{\mathrm{s}}$ & $-\log(0.8)/80640$ & 1/min\\
$\theta_{\SI}$ & 0.18 & ml/($\mu$U day)\\
$\tau_{\SI}$ & 150& day\\
$\zeta_3$ & 1.4& - \\
$k_{n\SI}$ & $5\times10^{6}/\lambda_t$ & (pg/ml) day\\
$\tau_{\Gamma}$ & 2.14& day\\
$\tau_{\Sigma}$ & 249.9& day\\
$\tau_{B}$ & 8570& day\\
$\upsilon$ & 5& litre\\
$\kappa$ & 700& 1/day\\
$\rho_0$ & 864& mg/(dl day)\\
$\omega$ & 70& kg \\
$\upsilon_g$ & 117 & dl\\
$\eta_0$ & 1.44& 1/day\\
\midrule
$\lambda_{\VL}/\lambda_t$ & $(\lambda_{\IL}/\kappa_{\mathrm{s}})/\lambda_t$ & (pg/ml) day\\
$\lambda_{\SI}$ & ${\SI}_0$ & ml/($\mu$U day)\\
$\lambda_{\Gamma}$ & $\gamma_{\mathrm{max}}$ & - \\
$\lambda_{\Sigma}$ & $\sigma_{\mathrm{ISRmax}}$ &$\mu$U/($\mu$g day)\\
$\lambda_{B}$ & $B_0$ & mg\\
$\lambda_{I}$ & $I_0$ & $\mu$U/ml\\
$\lambda_{G}$ & $G_0$ & mg/dl\\
$\lambda_{\Ie I}$ & $\lambda_{\Ie}/\lambda_I$ & - \\
$\rho_{\lambda}$ & $\rho_0/\lambda_G$ & 1/day \\
$\lambda_{tG}$ & $\lambda_t/(\lambda_G\upsilon_g)$ & min/(mg day) \\
$\lambda_{\SI I}$ & $\lambda_{\SI}\lambda_I$ & 1/day\\
\midrule
$\VL_0$ & 0 & (pg/ml) min\\
${\SI}_0$ & 0.8& ml/($\mu$U day) \\
$\Gamma_0$ & -0.00666& -\\
$\Sigma_0$ & 536.67&$\mu$U/ ($\mu$g day) \\
$B_0$ & 1000.423& mg \\
$I_0$ & 9.025& $\mu$U/ml \\
$G_0$ & 99.7604& mg/dl\\
$\Gamma_{0\lambda}$ & $\Gamma_0/\lambda_{\Gamma}$ & - \\
$\Sigma_{0\lambda}$ & $\Sigma_0/\lambda_{\Sigma}$ & - \\
\bottomrule
\end{tabular}
\end{table}
The units of some parameters in Table~\ref{table:params_long_term} are scaled with $\lambda_t$ in order to be in days.
Based on expert opinion, the initial conditions should be in the
following intervals:
${\SI}_0 \in [0, 0.8]$,
$\Gamma_0 \in [-0.1, 0.1]$,
$\Sigma_0 \in [3, 600]$,
$B_0 \in [0, 9000]$,
$I_0 \in [0, 100]$,
$G_0 \in [0,600]$.

\subsection{Model reduction}\label{appendix:reduction}

\subsubsection{Analytical solution of the short-term equations}

We calculate the analytical solution $\VO$ from System~\eqref{eq:y1} first in the interval $[0,\delta]$, where the control $u(t) = 1$, then use the value at $\delta$ as initial condition to solve the ODEs analytically in the interval $(\delta,\nu)$, where $u(t) = 0$:
\begin{lemma}
The solution to $\VO'(t) = \lambda_t \theta \big(1 - \VO(t)\big)$ for $0<t\leq\delta$ with initial condition $\VO(0) = 0$ is given by
\[
\VO(t) = 1 - \exp(- \lambda_t \theta t).
\]
\end{lemma}
Inserting $\delta$ yields
$\VO(\delta) = 1 - \exp(-\lambda_t \theta \delta)$, which is
the initial condition for the next interval:
\begin{lemma}
The solution to $\VO'(t) = - \lambda_t \theta \VO(t)$ 
for $\delta<t<\nu$ with initial condition 
$\VO(\delta) = 1 - \exp(-\lambda_t \theta \delta)$ is given by
\[
\VO(t) = \big(\exp(\lambda_t \theta \delta) -1\big)\exp(-\lambda_t \theta t).
\]
\end{lemma}
Calculating the analytical solution for the other
four state variables from
System~\eqref{eq:y1} ($\Gpr$, $\Gup$, $\Ie$, $\IL$) follows the same steps, additionally using the analytical solution of $\VO$.
For the sake of simplicity, we just illustrate it for $\Gpr$:
\begin{lemma}
The solution to $\Gpr'(t) = \lambda_t \alpha_2 \big(1 - \exp(-\lambda_t \theta t) -\Gpr(t)\big)$ for $0<t\leq\delta$ with
initial condition $\Gpr(0) = 0$ is given by
\[
\Gpr(t) = 1 + \frac{\alpha_2}{\theta - \alpha_2}\exp(-\lambda_t\theta t)
- \frac{\theta}{\theta-\alpha_2} \exp(-\lambda_t \alpha_2 t).
\]
\end{lemma}
Inserting $\delta$ yields the initial condition for the next interval:
\begin{lemma}
The solution to
\[
\Gpr'(t) = \lambda_t \alpha_2\Big(\big(\exp(\lambda_t\theta \delta) -1\big)\exp(-\lambda_t\theta t) - \Gpr(t)\Big)
\]
for $\delta<t<\nu$
with initial condition
\[
\Gpr(\delta) = 1 + \frac{\alpha_2}{\theta - \alpha_2}\exp(-\lambda_t\theta \delta)
- \frac{\theta}{\theta-\alpha_2} \exp(-\lambda_t \alpha_2 \delta)
\]
is given by
\begin{align*}
\Gpr(t) =& \frac{\theta}{\theta-\alpha_2}\big(\exp(\lambda_t \alpha_2\delta) - 1\big) \exp(-\lambda_t \alpha_2 t)\\ &- \frac{\alpha_2}{\theta - \alpha_2}\big(\exp(\lambda_t\theta \delta) -1\big)\exp(-\lambda_t\theta t).
\end{align*}
\end{lemma}

\subsubsection{Analytical average for the short-term equations}

To calculate the average $\mu_{\VO}$ of $\VO(t)$ in the interval $[0,\nu)$,
we start with
\begin{align*}
\int_0^\delta\VO(t) \,\mathrm{d}t &= \int_0^\delta1 - \exp(-\lambda_t \theta t) \,\mathrm{d}t
= \delta- \frac{1}{\lambda_t \theta}\big(1 - \exp(-\lambda_t \theta\delta)\big),\\
\int_\delta^\nu \VO(t) \,\mathrm{d}t &=
\frac{1}{\lambda_t \theta}\big(\exp(\lambda_t \theta \delta) - 1\big)\Big(
\exp(-\lambda_t \theta\delta)
-\exp(-\lambda_t \theta\nu)
\Big)\\
&= \frac{1}{\lambda_t \theta}\Big(1 - \exp(-\lambda_t \theta\delta)
+\exp(-\lambda_t \theta\nu)\big(1-\exp(\lambda_t \theta \delta)\big)
\Big).
\end{align*}
Combining these identities yields
\begin{align*}
\mu_{\VO} &= \frac{1}{\nu}\int_0^\nu \VO(t) \,\mathrm{d}t
= \frac{1}{\nu}\bigg(\int_0^\delta\VO(t) \,\mathrm{d}t + \int_\delta^\nu \VO(t) \,\mathrm{d}t\bigg)\\
&= \frac{1}{\nu}\bigg(
\delta
+\frac{1}{\lambda_t \theta}\Big(\exp(-\lambda_t \theta\nu)\big(1-\exp(\lambda_t \theta \delta)\big)
\Big)\bigg).
\end{align*}
Note that the result is approximately equal to
$\delta/\nu$, which is the average value of $u(t)$.\\
For $\Gpr$, it holds that
\begin{align*}
\int_0^\delta\Gpr(t) \,\mathrm{d}t
= \delta&- \frac{\theta}{\lambda_t \alpha_2(\theta-\alpha_2)}\big(1 - \exp(-\lambda_t \alpha_2\delta)\big)\\
&+\frac{\alpha_2}{\lambda_t \theta(\theta - \alpha_2)}\big(1 - \exp(-\lambda_t \theta\delta)\big),
\end{align*}
and
\begin{alignat*}{2}
\int_\delta^\nu \Gpr(t) \,\mathrm{d}t =
&\frac{\theta}{\lambda_t \alpha_2(\theta-\alpha_2)}
\Big(
&&1 - \exp(-\lambda_t \alpha_2\delta)\\
& &&+\exp(-\lambda_t \alpha_2\nu)\big(1-\exp(\lambda_t \alpha_2\delta)\big)\Big)\\
&-\frac{\alpha_2}{\lambda_t \theta(\theta - \alpha_2)}
\Big(
&&1 - \exp(-\lambda_t \theta\delta)\\
& &&+\exp(-\lambda_t \theta\nu)\big(1 -\exp(\lambda_t \theta\delta)\big)\Big).
\end{alignat*}
And hence,
\begin{alignat*}{2}
\mu_{\Gpr} &= \frac{1}{\nu}\bigg(
&&\int_0^\delta\Gpr(t) \,\mathrm{d}t + \int_\delta^\nu \Gpr(t) \,\mathrm{d}t\bigg)\\
& =\frac{1}{\nu}\bigg(
&&\delta
+ \frac{\theta}{\lambda_t \alpha_2(\theta-\alpha_2)}
\Big(\exp(-\lambda_t \alpha_2\nu)\big(1-\exp(\lambda_t \alpha_2\delta)\big)\Big)\\
& &&-\frac{\alpha_2}{\lambda_t \theta(\theta - \alpha_2)}
\Big(\exp(-\lambda_t \theta\nu)\big(1 -\exp(\lambda_t \theta\delta)\big)\Big)\bigg).
\end{alignat*}
The mean values for the state variables $\Gup$, $\Ie$ and $\IL$ are calculated analogously to $\Gpr$.

\subsection{Simulations for the approximation error analysis}\label{appendix:simulations}

To quantify the numerical error resulting from model reduction,
we varied six system parameters and three initial conditions,
as summarized in Table~\ref{table:varying_parameters}.
\begin{table}[htbp]
    \centering
    \caption{Parameters (top) and initial conditions (bottom) that are varied for the numerical evaluation of the approximation error, resulting in $3^9 = 19\,683$ simulations.}\label{table:varying_parameters}
    \begin{tabular}{r l}
        \toprule
        Parameter & Values \\
        \midrule
        $\nu$ & 2,4,6\\
        $\delta$ & 30/1440,45/1440,60/1440\\
        $\xi$ & 20,40,60\\
        $\theta_{\SI}$ & 0.18,0.28,0.38\\
        $\tau_{\SI}$ & 90,210,330\\
        $\omega$ & 50,90,130\\
        \midrule
        $B_0$ & 800,1000,1200\\
        $I_0$ & 5,10,15\\
        $G_0$ & 70,90,110\\
        \bottomrule
    \end{tabular}
\end{table}
Each parameter configuration is denoted by
$\Theta_i$, where $i = 1,\dots,19\,683$.
Both systems are simulated for $t_\mathrm{end} = 1824$ days (corresponding to 5 years)
and the solutions are compared every 2 days, i.e., at $t_m \isdef 2m$, where $m = 1, \dots, 912$.
To quantify the approximation error, two measures are used.
The maximal deviation between the two models over the entire time horizon under consideration,
\[
E_i^{\mathrm{basal}} \isdef \max_{m = 1, \dots, 912}
|\mathbf{y_2}(t_m,\mathbf{y_1},\Theta_i) - \widehat{\mathbf{y_2}}(t_m,\mathbf{\mu},\Theta_i)|
\in \mathbb{R}^{7},
\]
and the point-wise error after 5 years,
\[
E_i^{\mathrm{5y}} \isdef 
\big(\mathbf{y_2}(t_\mathrm{end},\mathbf{y_1},\Theta_i) - \widehat{\mathbf{y_2}}(t_\mathrm{end},\mathbf{\mu},\Theta_i)\big)\in \mathbb{R}^{7}.
\]
\printbibliography
\end{document}